\documentclass[copyright,creativecommons]{eptcs}
\providecommand{\event}{TARK 2021}

\usepackage{breakurl}

\usepackage{colonequals}
\usepackage{etoolbox}
\usepackage{amssymb}
\usepackage{amsmath}
\usepackage{amsthm}

\usepackage{textcomp}

\usepackage{mathrsfs}
\usepackage{paralist}
\usepackage{tikz}
\usepackage{tikz-cd}

\newtheorem{theorem}{Theorem}
\newtheorem{lemma}[theorem]{Lemma}
\newtheorem{corollary}[theorem]{Corollary}
\newtheorem{definition}[theorem]{Definition}

\newtheorem{remarque}[theorem]{Remark}

\newcommand{\bbbt}{\mathbb{T}}
\newcommand{\bbbn}{\mathbb{N}}
\newcommand{\ce}{\colonequals}
\newcommand{\cce}{\coloncolonequals}

\newcommand{\B}[2]{B_{#1}{#2}}
\renewcommand{\H}[2]{H_{#1}{#2}}

\newcommand{\CH}[2]{C^{\Diamond H}_{#1}{#2}}

\newcommand{\agents}{\mathcal{A}}

\newcommand{\globalstates}{\mathscr{G}}

\newcommand{\correct}[1]{correct_{#1}}

\newcommand{\trueoccurred}[2][]{\overline{\mathit{occurred}}_{#1}(#2)}

\newcommand{\eventually}[1]{\lozenge{#1}}
\newcommand{\always}[1]{\square{#1}}
\newcommand{\kstruct}[3]{(\Kstruct{#1},{#2},{#3})}
\newcommand{\START}{\operatorname{START}}
\newcommand{\FIRE}{\operatorname{FIRE}}
\newcommand{\Y}[1]{Y{#1}}

\newcommand{\truefire}[1][]{\overline{fire}_{#1}}

\newcommand{\truestart}[1][]{\overline{start}_{#1}}

\newcommand{\envprotocol}[1]{P_{\epsilon}\ifstrempty{#1}{}{\left(#1\right)}}
\newcommand{\agprotocol}[2]{{P_{#1}\ifstrempty{#2}{}{\left(#2\right)}}}

\newcommand{\points}{{\agents\times\bbbt}}

\newcommand{\run}[3][]{r#1_{#2}\left(#3\right)}

\newcommand{\System}[1][]{R^{#1}}

\newcommand{\system}[1]{{R^{#1}}}

\newcommand{\Kstruct}[1]{I^{#1}}

\newcommand{\K}[2]{K_{#1}\ifstrempty{#2}{}{#2}}

\newcommand{\intsys}{I}

\newcommand{\prop}{\mathit{Prop}}

\title{Fire!\thanks{Funded by the Austrian Science Fund (FWF) project ByzDEL P33600.}}
\author{Krisztina Fruzsa\thanks{PhD student in the FWF doctoral program LogiCS (W1255).} \qquad\qquad Roman Kuznets \qquad\qquad Ulrich Schmid
\institute{TU Wien}
\email{$\{$kfruzsa,rkuznets,s$\}$@ecs.tuwien.ac.at}
}
\def\titlerunning{Fire!}
\def\authorrunning{K. Fruzsa, R. Kuznets \& U. Schmid}

\begin{document}
\maketitle

\begin{abstract}
In this paper, we provide an epistemic analysis of a simple variant of 
the fundamental consistent broadcasting primitive for byzantine fault-tolerant 
asynchronous distributed systems. Our Firing Rebels with Relay~(FRR) primitive enables
agents with a local preference for acting/not acting to trigger an action~(FIRE) at all correct agents, in an all-or-nothing fashion. By using the
epistemic reasoning framework for byzantine multi-agent systems introduced
in our TARK'19 paper, we develop the necessary and 
sufficient state of knowledge that needs to be acquired by the agents in order
to FIRE. It involves eventual common hope (a modality related to belief),
which we show to be attained already by achieving eventual mutual hope
in the case of~FRR. We also identify subtle variations of the necessary and 
sufficient state of knowledge for~FRR for different assumptions on the 
local preferences.
\end{abstract}

\section{Motivation and Background}
\label{sec:intro}

In their PODC'18 paper ``Silence''~\cite{GM18:PODC}, Goren and Moses introduced and epistemically analyzed
\emph{silent choirs} as a fundamental primitive for message-optimal protocols in 
synchronous fault-tolerant distributed systems where computing nodes (agents\footnote{Since distributed systems are just one instance of multi-agent systems, we will use the term ``agent'' instead of ``process.''}) can crash.
In synchronous systems, where one can time-out messages, it is well-known~\cite{Lam78} 
that an agent can convey information also by \emph{not} sending some message. In a system
where the sender may also crash, however, the receiver cannot infer this information 
from not receiving the message. Still, if only up to~$f$ of the $n > f$~agents in a 
system may crash, a silent choir of $f+1$~agents that convey identical information 
suffices: at least one agent in the choir must be correct, so its silence can 
be relied on.

Whereas silent choirs also work  in systems where the faulty agents may behave 
arbitrarily (byzantine~\cite{lamport1982byzantine}),
the problem of not conveying information faithfully now also plagues messages that
\emph{are} sent, as they could originate from a faulty sender or forwarding agent.
In this paper, we will introduce and epistemically analyze a fundamental primitive
\emph{Firing Rebels with Relay}~(FRR), which nicely captures exactly these issues. It is
a simplified version of the \emph{consistent broadcasting} primitive introduced by
Srikanth and Toueg in~\cite{ST87}, which has been used as a pivotal building block 
in distributed algorithms for byzantine fault-tolerant clock synchronization~\cite{DFPS14:JCSS,FS12:DC,RS11:TCS,ST87,WS09:DC} and synchronous consensus~\cite{ST87:abc}, for example.\looseness=-1

Informally, FRR~requires that \emph{every} correct agent perform an action called $\FIRE$, in an 
all-or-none fashion (though not necessarily simultaneously), and only if at least one correct agent locally observed a trigger event called $\START$. Note that we have replaced the need to broadcast explicit
information by just triggering an action, which makes~FRR essentially a non-synchronous variant of 
the Firing Squad problem~\cite{BL87}, hence its name.
In crash-prone systems, FRR~is trivial to solve, even for large~$f$: Indeed, every agent who
observes $\START$ or receives a notification message (for the first time) just invokes 
$\FIRE$ and sends a notification message to everyone. This guarantees that if a single 
correct agent observes $\START$, every correct agent will invoke $\FIRE$ 
(agents that crash during the run may or may not issue $\FIRE$ here). Observe that
this solution involves a trivial silent choir, namely, when no agent observes $\START$.\looseness=-1

In the presence of byzantine agents, however, this solution does not work, as faulty agents 
may send a notification without having observed anything. A correct solution for~FRR must, hence, prevent the faulty agents from triggering $\FIRE$ at any correct agent.
In this paper, we will establish the necessary and sufficient state of knowledge for correctly solving~FRR 
in our epistemic reasoning framework for byzantine multi-agent systems~\cite{KPSF19:TARK,KPSF19:FroCos,PKS19:TR}. 
At least since the ground-breaking work by Halpern and Moses~\cite{HM90}, the knowledge-based approach~\cite{bookof4} has been known as a powerful tool for analyzing distributed systems. In a nutshell, it uses epistemic logic~\cite{Hin62} to reason about knowledge and belief in distributed systems. As agents take actions (e.g.,~$\FIRE$) based on the accumulated local knowledge, reasoning about the latter is useful both for protocol design and impossibility proofs.

In the \emph{runs-and-systems} framework for reasoning about multi-agent systems~\cite{bookof4,HM90}, the set of all possible runs~$r$ (executions) of a system~$I$ determines the Kripke model, formed by pairs~$(r,t)$ of a run~$r\in I$ and time~$t\in\mathbb{N}$ representing global states~$r(t)$. 
Note that time is modeled as discrete for simplicity, without necessarily being available to the agents. 
Two pairs~$(r,t)$~and~$(r',t')$ are indistinguishable for agent~$i$ if{f} $i$~has the same local state in both global states represented by those points, formally, if $r_i(t)=r_i'(t')$. 
A modal \emph{knowledge operator}~$K_i$ is used to capture that agent~$i$ knows some fact~$\varphi$ in run~$r \in I$ at time~$t\in\mathbb{N}$. 
Formally, $(I,r,t) \models K_i\varphi$ if{f} for every~$r'\in I$ and for every~$t'$ with $r_i(t)=r_i'(t')$ it holds that $(I,r',t')~\models~\varphi$. 
Note that $\varphi$~can be a formula containing arbitrary atomic propositions like $\trueoccurred{e}$ (event~$e$ occurred) or $\correct{i}$ ($i$~did not fail yet), as well as other knowledge operators and temporal modalities like $\lozenge$~(eventually) and $\Box$~(always), combined by standard logical operators~$\neg, \wedge, \vee$,~and~$\to$.  
For example, $(I,r,t)~\models~\eventually{K_i\trueoccurred{e}}$ states that there is some time $t'\geq t$ when $i$~knows that event~$e$ occurred. 
Important additional modalities for a group~$G$ of agents are \emph{mutual knowledge} $E_G\varphi \ce \bigwedge_{i\in G} K_i\varphi$ and \emph{common knowledge}~$C_G\varphi$ that can be informally expressed as an infinite conjunction $C_G\varphi \equiv E_G\varphi \wedge E_G(E_G\varphi) \wedge \dots$; in other words, this means that every agent in~$G$ knows~$\varphi$, and every agent in~$G$ knows that every agent in~$G$ knows~$\varphi$, and so on.\looseness=-1

Actions performed by the agents when executing a protocol take place when they have accumulated some specific epistemic knowledge. According to the pivotal \emph{Knowledge of Preconditions Principle}~\cite{Mos15TARK}, it is universally true that if $\varphi$~is a necessary condition for agent~$i$ to take a certain action then $i$~may act only if $K_i \varphi$~is true. 
For example, in order for agent~$i$ to decide on~0 in a binary consensus algorithm, $i$~must know that some agent~$j$ has started with initial value $x_j=0$, i.e., $K_i(\exists j: x_j=0)$ must hold true. 
Showing that agents act without having attained the respective necessary knowledge is, hence, a very effective way for proving incorrectness of protocols. Conversely, optimal distributed algorithms can be designed by letting agents act as soon as all respective necessary knowledge has been established. 
Prominent examples are the protocols based on silent choirs analyzed in~\cite{GM18:PODC} and the \emph{unbeatable} consensus protocols introduced in~\cite{CGM14}, which are not just worst-case optimal but also  not strictly dominated w.r.t.~termination time by any other protocol in \emph{any} execution.

\medskip
\noindent
\textbf{Related work:} The knowledge-based approach has been used for studying several distributed computing problems in systems with uncertainty but no failures. In~\cite{ben2014beyond}, Ben-Zvi and Moses considered the simple \emph{ordered response} problem in distributed systems, where the agents had to respond to an external $\START$ event by executing a special one-shot action $\FIRE$  in a given order~$i_1,i_2,\dots$. The authors showed that, in every correct solution, agent~$i_k$ has to establish nested knowledge $K_{i_k}K_{i_{k-1}}\dots K_{i_1}\trueoccurred{\START}{}$ before it can issue $\FIRE$ and that this nested knowledge is also sufficient.  
In the conference version~\cite{BM10:DISC} of~\cite{ben2014beyond}, the authors also considered the \emph{simultaneous  response} problem where all agents had to issue $\FIRE$ at the same time. It requires the group~$G$ of firing agents to establish common knowledge $C_G\trueoccurred{\START}{}$~\cite{HM90}. This work was later extended to responses that are not simultaneous but tightly coordinated in time~\cite{BM13:ICLA,GM13:TARK}.

The knowledge-based approach has also been successfully applied to fault-tolerant synchronous distributed systems. Agents suffering from crash or omission failures have been studied in~\cite{moses1986programming,MT88}, primarily in the context of agreement problems~\cite{dwork1990knowledge,halpern2001characterization},  which require some form of common knowledge. Important ingredients here are the indexical set of correct agents and a related belief operator $B_i\varphi \ce K_i(\correct{i} \rightarrow \varphi)$~\cite{MosSho93AI}, which states that agent~$i$ knows~$\varphi$ to be true in all runs where $i$~is correct. 
This notion of ``defeasible knowledge'' also underlies a variant
of common knowledge that has been used successfully for characterizing
simultaneous distributed agreement~\cite{moses1986programming,MT88}. Closer related to
our FRR~problem is eventual distributed agreement studied in~\cite{halpern2001characterization}, where the stronger notion of continual common knowledge proved its value.
The latter needs to hold throughout a run, i.e.,~from the beginning,
which makes sense here since it is only applied to conditions on the
initial state. Continual common knowledge does not seem
readily applicable to~FRR, however, as $\START$ can occur at any time in a run.
More recent results are the already mentioned unbeatable consensus algorithms in synchronous systems with crash failures~\cite{CGM14} and the silent-choir based message-optimal protocols~\cite{GM18:PODC}.

\medskip
\noindent  
\textbf{Detailed contributions:} 
We rigorously define the FRR~problem and its weaker variant~FR, without the all-or-nothing requirement (agreement), 
in epistemic terms and identify the necessary and sufficient state of knowledge that must be established by a correct agent in order to issue $\FIRE$ in every correct solution for~FRR. 
Since FRR~involves distributed agreement, the required state of knowledge involves some form of (eventual) common knowledge of $\trueoccurred{\START}{}$. 
Interestingly, it turned out that establishing the respective eventual mutual knowledge (namely, ``eventual mutual hope'' where the hope modality is defined as $H_i\varphi \ce \correct{i} \to B_i\varphi$) already implies the required common knowledge (namely, ``eventual common hope''). 
We also identify subtle variations of the necessary and sufficient state of knowledge for~FRR for different assumptions on the occurrence of $\START$.

Whereas identifying the necessary and sufficient state of knowledge for the agents to $\FIRE$
does not immediately lead to efficient practical protocols, it is an important first step towards
this goal. Indeed, as for the
ordered response problem in~\cite{ben2014beyond}, for example, we expect this knowledge to
lead  to necessary and sufficient \emph{communication structures}, which must be present 
in every run of any correct protocol solving~FRR. 
Knowing the latter would not only enable us to decide right away whether the communication guarantees provided by some distributed system 
allow to solve~FRR, but also facilitate the design of efficient protocols.

\medskip
\noindent
\textbf{Paper organization:}
In Section~\ref{sec:model}, we introduce some minimal basic notation from our modeling framework~\cite{PKS19:TR}. In Section~\ref{sec:FRR}, we provide the detailed definition and an epistemic analysis of the FRR~problem. Some
conclusions and directions of future work in Section~\ref{sec:conclusions} complete our paper.

\section{Preliminaries}
\label{sec:model}

In this section, we outline the basic concepts and facts that are employed in our epistemic analysis of Firing Rebels with Relay~(FRR). Our analysis was actually performed within the rigorous framework (that is based on the standard runs-and-systems framework) first developed in~\cite{PKS19:TR}, which incorporates agents' ability to arbitrarily deviate from normative behavior. This framework enables one to formally express the epistemic limitations (of the agents) that the presence of possibly fully byzantine agents in the system imposes.

In particular, we proved in~\cite{KPSF19:FroCos} that asynchronous agents in a system with fully byzantine agents can never \emph{know} that a particular event took place or even that an agent performed a particular action (since the local state of a malfunctioning agent may have been corrupted, the above statement applies even to the agent's knowledge of its own actions). This rather disappointing result stems from the inability of even correct agents to exclude the possibility of the so-called \emph{brain-in-a-vat scenario}. In other words, an agent can never be sure that the events and actions recorded in its local history truly happened as recorded rather than being figments of its own malfunction.
To make matters worse, agents can also never \emph{know} that they are correct either. 
Thus, unable to rely on knowledge or their own correctness, our agents are forced to rely instead on \emph{belief} $B_i\varphi \ce K_i(\correct{i} \rightarrow \varphi)$~(cf.~\cite{MosSho93AI}). In this paper, we show that even belief is not always appropriate and needs to be replaced with a modality we called \emph{hope} defined as $H_i\varphi \ce \correct{i} \to B_i\varphi$ (for a detailed explanation see Remark~\ref{rem:hope_orig}).

Since the epistemic analysis presented in this paper is protocol-independent and does not rely on the artefacts of our modeling, we omit all details  irrelevant to the task at hand and present our findings in an epistemic language with temporal modalities that is interpreted in Kripke models generated by runs in our framework. The purpose of this section is to provide all the necessary ingredients (referring the reader to~\cite{KPSF19:TARK,KPSF19:FroCos,PKS19:TR} for full details of the said framework).

We fix a finite set $\agents=\{1,\dots,n\}$ of \emph{asynchronous agents} with \emph{perfect recall}. 
Each agent~$i\in\agents$ can perform \emph{actions} (according to its protocol), e.g.,~send \emph{messages}. One of the actions any agent can perform is $\FIRE$. Agents also witness \emph{events} (triggered by the \emph{environment}) such as message delivery. One of the events that can be observed by any agent is $\START$.
We use a discrete time model governed by a global clock with domain~$\bbbt = \bbbn$.
All events taking place after clock time~$t\in\bbbt$ and no later than~\mbox{$t+1$} 
are grouped into a \emph{round} denoted~\mbox{$t+{}$\textonehalf{}} and are treated as  
happening simultaneously.
Apart from actions, everything in the system is governed by the \emph{environment}.
Unlike the environment, agents only have limited local information, in particular, being asynchronous, do not have access to the global clock. This is achieved by allowing them not to perform actions in some rounds and allowing them, in the absence of either actions or events, to stay in the same local state for several rounds in a row. The agents have perfect recall in the sense that, once recorded in their local history, actions and events are never forgotten.

No assumptions apart from liveness are made about the communication. Messages can be lost, arbitrarily delayed, and/or delivered in the wrong order. In addition, the environment may cause at most $f$~agents to become \emph{byzantine} faulty. A byzantine faulty agent can perform any action irrespective of its protocol and ``observe'' events that did not happen. It can also have false memories about actions it has performed. At the same time, like the global clock, such malfunctions are not directly visible to an agent.\looseness=-1

Throughout the paper, horizontal bars signify phenomena that are correct. Note that   the absence of this bar should not be equated to faultiness but rather  means the absence of a claim of correctness.

Agent~$i$'s local view of the system immediately after round~\mbox{$t+{}$\textonehalf}, referred to as (\emph{process-time} or \emph{agent-time}) \emph{node}~$(i,t+1)$, is recorded in $i$'s~\emph{local state} $r_i(t+1)$, also called $i$'s~\emph{local history}. 
 A \emph{run}~$r$ is a sequence of \emph{global states} $r(t) = \bigl(r_{\epsilon}(t),r_1(t),\dots,r_n(t)\bigr)$ of the whole system consisting of the \emph{state~$r_\epsilon(t)$ of the environment} and local states~$r_i(t)$ of every agent. Unlike local states, the global state of the system necessarily updates every round to include all actions and events that happened (even the empty set thereof is faithfully recorded and modifies the global state).
  The set of all global states is denoted~$\globalstates$.
 
What happens in each round is determined by nondeterministic protocols~$\agprotocol{i}{}$  of the agents, the nondeterministic protocol~$\envprotocol{}$ of the environment, and chance, the latter implemented as the \emph{adversary} part of the environment (the exact technical details are not important for this paper). 
 
In our epistemic analysis, we consider pairs~$(r,t)$  of a run~$r$  and time~$t$. A \emph{valuation function}~$\pi $ determines  whether an atomic proposition from~$\prop$ is true in run~$r$ at time~$t$. The determination is arbitrary except for a small set of \emph{designated atomic propositions} 
whose truth value at~$(r,t)$ is fully determined by the state of the system. More specifically, for~$i\in\agents$~and~$t\in \bbbt$,

$\correct{i}$ is true at~$(r,t)$ if{f} no faulty event happened to~$i$  by time~$t$;\looseness=-1

$\trueoccurred[i]{o}$ is true at~$(r,t)$ if{f} $i$'s~local history~$r_i(t)$ contains an \emph{accurate} record of action/event~$o$ occurring (for example, in this paper we use action $o=\FIRE$ and event $o=\START$); 
   
$\trueoccurred{o} \ce \bigvee\nolimits_{i \in \agents}\trueoccurred[i]{o}$.

An \emph{interpreted system} is a pair $I = (R, \pi)$ where $R$~is the set of considered runs. 
The  language is $\varphi \cce p \mid \lnot \varphi \mid (\varphi \land \varphi) \mid K_i \varphi \mid \lozenge \varphi \mid Y \varphi$ where $p \in \prop$ and~$i\in\agents$; derived Boolean connectives are defined in the usual way; $\Box \varphi \ce \lnot \Diamond \lnot \varphi$.  
Truth  for these \emph{formulas} is defined in the standard way, in particular, for a run~$r \in R$, time~$t \in \bbbt$, atomic proposition $p \in \prop$, agent~$i \in \agents$, and formula~$\varphi$, we have $(I,r,t) \models p$ if{f} $(r,t) \in \pi(p)$, and $(I,r,t) \models K_i \varphi$ if{f} $(I,r',t') \models \varphi$ for any~$r'\in R$~and $t' \in \bbbt$~such that $r_i(t) = r'_i(t')$,  and $(I,r,t) \models \lozenge \varphi$ if{f} $(I,r,t') \models \varphi$ for some~$t' \ge t$, and $(I,r,t) \models Y\varphi$ if{f} $t>0$ and $(I,r,t-1) \models \varphi$ . 
A formula~$\varphi$ is valid in~$I$, written $I \models \varphi$, if{f} $(I,r,t) \models \varphi$ for all $r \in R$ and $t \in \bbbt$.\looseness=-1 

\section{The Firing Rebels Problem}
\label{sec:FRR}

In~\cite{ST87}, Srikanth and Toueg introduced the consistent broadcasting primitive, which proved its value in several different contexts, ranging from low-level byzantine fault-tolerant tick generation in various system models~\cite{DFPS14:JCSS,FS12:DC,RS11:TCS,WS09:DC} to classic clock synchronization~\cite{ST87} to Byzantine Agreement~\cite{ST87:abc}. As argued already in Section~\ref{sec:intro}, it gave rise to our \emph{Firing Rebels with Relay} problem~FRR~\cite{Fim18:master}, which can be seen as a natural generalization of a \emph{silent choir} in byzantine fault-tolerant systems~\cite{GM18:PODC}.
As an important building block for designing byzantine fault-tolerant systems, it is therefore a natural target for a detailed epistemic analysis in our framework~\cite{PKS19:TR}.

Our Firing Rebels problems assume that every agent~$i \in \agents$ may observe an event $\START$ and may generate an action $\FIRE$ according to the following specification:

\begin{definition}[Firing Rebels with and without Relay]\label{def:FR}
A system is consistent with \emph{Firing Rebels~(FR)} for~$f\geq 0$ when all runs 
satisfy:
	\begin{itemize}
		\item[\textup{(C)}] \emph{Correctness}: If at least $2f+1$~agents learn that\/ $\START$ occured at a correct agent, all correct agents perform\/ $\FIRE$ eventually.
		\item[\textup{(U)}] \emph{Unforgeability}: If a correct agent performs\/ $\FIRE$, then\/ $\START$ occurred at a correct agent.
	\end{itemize}
Moreover, the system is consistent with \emph{Firing Rebels with Relay~(FRR)} if every run also satisfies:
	\begin{itemize}
		\item[\textup{(R)}] \emph{Relay}: If a correct agent performs\/ $\FIRE$, all correct  agents perform\/ $\FIRE$ eventually.
	\end{itemize}
\end{definition}

\begin{remarque}[Variants of Correctness]\label{rem:correctness}
A different specification for Correctness can sometimes be found in literature: ``\emph{If at least $f + 1$~reliable agents locally observed $\START$, then some reliable agent fires eventually}'' (see,~e.g.,~\cite{BL87}). 
Here, a reliable agent is one that \emph{will always} follow its protocol, which corresponds to a forever correct agent in our terminology. 
In the case of~FRR, by invoking\/~\textup{(R)}, this specification implies ``\emph{If at least $f + 1$~reliable agents locally observed $\START$, then all reliable agents fire eventually}.'' 
Relying on such a specification in asynchronous settings is problematic, however, because reliability depends on the future behavior of the system. Even complete knowledge of the global state, at a given time in a run, does not allow to identify the reliable agents whose  observations of\/ $\START$ could be relied upon.
Thus, we require $2f+1$~arbitrary (correct or faulty) agents instead. Of course,
given the limit of $f$~faulty agents per run, at least $f+1$~(not necessarily the same)~of these agents will 
remain reliable in every run.
Moreover, we relax the condition of the $2f+1$~agents locally observing\/ $\START$ to each of them learning that\/  $\START$ happened to some correct agent. This is preferable, because direct observation is only one possible way of ascertaining that\/ $\START$ occurred. For instance, if an agent has already  determined\footnote{Strictly speaking, the agent in this situation \emph{does not know} that the $f$~agents are faulty, but rather that they are faulty if it itself is not. By the same token, whenever we say ``learned,'' ``determined,'' or ``ascertained'' above, what we mean is reasoning under the assumption of its own correctness, i.e., the belief modality~$B_i$ rather than the knowledge modality~$K_i$.} who the $f$~faulty agents are, e.g.,~due to their erratic behavior in the past, then  a  confirmation of\/ $\START$ from just one other agent would be sufficient. 
\end{remarque}

\noindent We use the following abbreviations:
\begin{align*}
\noindent B_i \varphi &\ce K_i (\correct{i} \to \varphi)
&
H_i \varphi &\ce \correct{i} \to B_i \varphi = \correct{i} \to K_i (\correct{i} \to \varphi)
\\
E^B \varphi &\ce \bigwedge\nolimits_{j\in \agents} B_j \varphi 
&
E^H \varphi &\ce \bigwedge\nolimits_{j\in \agents} H_j \varphi 
\\
E^{\eventually B} \varphi &\ce \bigwedge\nolimits_{j\in \agents} \eventually B_j \varphi 
&
E^{\eventually H} \varphi &\ce \bigwedge\nolimits_{j\in \agents} \eventually H_j \varphi 
\end{align*} 
It has been shown in~\cite{Fru19:ESSLLI} that hope is a normal modality, in particular, $\models H_i(\varphi \land \psi) \to H_i\varphi \land H_i\psi$.
We define eventual common hope~$C^{\eventually H} \varphi$ as the greatest fixed point of the equation $\chi\leftrightarrow E^{\eventually H}(\varphi \land \chi)$ in the standard way (using the Knaster--Tarski theorem~\cite{KnasterTarski}) and use the following properties (the general versions of which can be found  in Lemma 11.5.7~in~\cite{bookof4}): for any interpreted system~$I$,
\begin{gather}
\label{eq:fixpointaxiom}
I \quad\models\qquad C^{\eventually H} \varphi \leftrightarrow E^{\eventually H}(\varphi \land C^{\eventually H} \varphi);
\\
\label{eq:inductionrule}
\textnormal{if} \qquad I\quad \models \quad\psi \to E^{\eventually H}(\varphi \land \psi), \qquad \textnormal{then} \qquad I \quad\models\quad \psi \to C^{\eventually H} \varphi.
\end{gather}

\subsection{Modeling}

\begin{definition}\label{def:atomicpropsFR}
	For an agent~$i \in \agents$, we define:
	\begin{align*}
	\truestart[i] &\quad\ce \quad\Y{\trueoccurred[i]{\START}} \wedge \correct{i}
	&
		\truefire[i] &\quad\ce\quad  \trueoccurred[i]{FIRE} \wedge \correct{i}
		\\
	\truestart[] &\quad\ce\quad \bigvee\nolimits_{j \in \agents} \truestart[j]
	&
		\truefire[] &\quad\ce\quad \bigvee\nolimits_{j \in \agents} \truefire[j]
	\end{align*}
\end{definition}
Note  that for one of these formulas to be true, it is necessary for (one of) the involved agent(s) to be correct not only at the time  the event/action in question occurred but also at the time of the evaluation. Using the yesterday modality ~$Y$ in $\truestart[i]$ accounts for the fact that agents cannot act on a precondition in the same round it is established.\looseness=-1

Using Def.~\ref{def:atomicpropsFR}, we can translate the specification of~FRR (stated in Def.~\ref{def:FR}) as follows:
\begin{definition}[Modeling Firing Rebels]\label{def:FRmodel}
	An interpreted system~$I$ is consistent with Firing Rebels with Relay for~$f\geq 0$ if the following conditions Correctness\/~\textup{(C)}, Unforgeability\/~\textup{(U)}, and Relay\/~\textup{(R)} hold:
	\begin{align*}
	\textup{(C)} \qquad& I \qquad\models  \qquad\bigvee\limits_{\parbox{1.3cm}{\scriptsize\centering$G\subseteq\agents$\\ $|G|\! =\! 2f\!+\!1$}} \bigwedge\limits_{j\in G} \K{j}{(\correct{j} \rightarrow \truestart[])} \rightarrow {\bigwedge\limits_{i\in \agents} \eventually (\correct{i} \to \truefire[i]})
	\\
	\textup{(U)} \qquad& I\qquad \models\qquad  \truefire[] \rightarrow \truestart
	\\
	\textup{(R)} \qquad& I \qquad\models\qquad \truefire[] \rightarrow {\bigwedge_{i\in \agents} \eventually (\correct{i}\to \truefire[i])}
	\end{align*}
	
\end{definition}

\begin{remarque}[Variants of eventuality]
The phrase \emph{all correct agents fulfill~$\varphi_i$ eventually} in Def.~\ref{def:FR} can be formalized in two different ways:
\begin{itemize}
\item ${\bigwedge_{i\in \agents} \eventually{(\correct{i} \to \varphi_i})}$ states that each agent will either become faulty at some point in the future or will fulfill its respective~$\varphi_i$ at some point in the future.
\item $\eventually{\bigwedge_{i\in \agents} (\correct{i} \to \varphi_i})$ states that there is one moment in the future by which every agent still correct fulfills its respective~$\varphi_i$.
\end{itemize}
The second statement is a strengthening of the first by demanding agents to have  one common moment  by which all correct agents fulfill their respective $\varphi_i$'s. On the other hand, the first variant is a more intuitive reading that is more widely applicable. 
Fortunately, for our model of~FRR with $\varphi_i = \truefire[i]$,  the two formulations are equivalent because, due to agents having perfect recall, $\correct{i} \to \truefire[i]$ is a stable fact.\looseness=-1
\end{remarque}

\subsection{Necessary and Sufficient Level of Knowledge}

The goal of this subsection is to
\begin{enumerate}[(a)]
\item
strengthen the given necessary conditions on a single agent's firing --- namely, that $\truestart$ must hold by Unforgeability~(U)  and  $\bigwedge_{i\in \agents} \eventually (\correct{i}\to \truefire[i])$ must hold by Relay~(R) --- 
to statements that describe the state of knowledge necessary for the agent to achieve before firing;
\item
show that firing upon reaching this state of knowledge is sufficient for satisfying the conditions Unforgeability~(U) and Relay~(R) on all correct agents eventually firing;
\item
 show how Correctness~(C) helps simplify these necessary and sufficient conditions in the presence of sufficiently many agents.
\end{enumerate}

Thus, protocols prescribing an agent to fire as soon as this state of knowledge is achieved are correct and optimal in the sense that firing earlier would violate the necessary conditions whereas firing prescribed by this state of knowledge is guaranteed to fulfill all requirements of~FRR. 

Note that the case when  insufficiently many agents learn that START occurred at a correct agent trivially satisfies condition~(C). In this case, FRR~reduces to~(U)+(R), a problem with a trivial solution of all correct agents not firing. It is the combination of all all three conditions that makes~FRR a problem worth the analysis.

The first lemma formalizes the fact that, since agents have perfect recall of their past perceptions, reasoning under the assumption of their own correctness leads them to \emph{believe} that these perceptions were accurate. For instance, an agent who recalls observing $\START$ believes that, unless it is faulty, a correct agent (namely, itself) observed $\START$. 
(A formal  proof can be found in the Appendix on p.~\pageref{proof:lemmasix}.)

\begin{lemma}
\label{lem:aux}
For any interpreted system~$I$ and any agent~$i \in \agents$:
\begin{align}
I &\quad\models\quad \truefire[i] \to B_i\truefire[i] 
\\
\label{eq:aux}
I &\quad\models\quad \truefire[i] \to B_i\truefire
\\
I &\quad\models\quad \truestart[i] \to B_i \truestart[i] 
\\
I &\quad\models\quad \truestart[i] \to B_i \truestart
\end{align}
\end{lemma}

Unforgeability~(U) states that $\truestart[]$ is a necessary condition for a correct agent firing. It follows from the Knowledge of Preconditions Principle that any correct agent must ascertain $\truestart[]$ (modulo its own correctness) before firing. We formalize this argument and provide an independent proof:
\begin{lemma}[State of knowledge necessary for firing in presence of Unforgeability~(U)]\label{lem:U}
	Let $I$~be an interpreted system consistent with Unforgeability\/~\textup{(U)}. For any agent~$i \in \agents$,
	\begin{equation}
	\label{eq:U:nec}
	 I \quad \models\quad \truefire[i]
	\to B_i \truestart[] .
	\end{equation}
\end{lemma}
\begin{proof}
Immediately follows from~\eqref{eq:aux}, (U),~and monotonicity/normality of~$B_i$.
\end{proof}

\begin{corollary}
For any interpreted system consistent with~FR,~\eqref{eq:U:nec}~is satisfied for all agents.
\end{corollary}

Similarly, lifting the Relay condition~(R) to the level of agent's knowledge yields the requirement that, in order to fire, an agent must believe that all correct agents eventually will have fired.
\begin{lemma}[State of knowledge necessary for firing in presence of  Relay~(R)]\label{lem:R}
	Let $I$~be an interpreted system consistent with Relay\/~\textup{(R)}. For any agent~$i \in \agents$,
\begin{equation}
	\label{eq:R:nec}
	I \quad\models\quad \truefire[i]
	 \to B_i{\bigwedge\nolimits_{j\in\agents} \eventually (\correct{j} \rightarrow \truefire[j])}.
\end{equation}
\end{lemma}
\begin{proof}
Immediately follows from~\eqref{eq:aux}, (R),~and monotonicity of~$B_i$.
\end{proof}

Combining the conditions  necessary  for~(U)~and~(R), we  establish the following level of  knowledge necessary for firing in~FRR (a proof can be found in the Appendix on p.~\pageref{proof:theoremten}):\looseness=-1

\begin{theorem}[State of knowledge necessary for firing in presence of both~(U)~and~(R)]\label{thm:necessaryFR}
Let $I$~be an interpreted system consistent with\/~\textup{(U)}\/~and\/~\textup{(R)}. For any agent~$i\in\agents$,
\[
I \qquad\models\qquad \truefire[i] \to  B_i\left(\truestart[] \wedge {E^{\eventually H}\truestart[]}\right).
\]
\end{theorem}

\begin{remarque}[Emergence of hope]\label{rem:hope_orig}
Note that $I \models \truefire[i] \to B_i{E^{\eventually B} \truestart}$ does not generally hold. 
We cannot strengthen the necessary condition  by replacing eventual mutual hope with eventual mutual belief, i.e.,~by omitting $\correct{j}$ therein. 
In other words, the use of hope for deeper iterations of knowledge modalities is crucial for the correct formulation. 
Indeed, in the case of our notion of belief, agent~$i$ can rarely have unconditional beliefs about another agent~$j$'s~\emph{beliefs}. 
The problematic situation is when  agent~$j$'s~perception is compromised. 
In that case, agent~$i$ has no way of ascertaining what $j$'s~erroneous input data might be and, hence, cannot determine what a correct agent would have inferred from these incorrect inputs. 
According to our notion of belief, whether agent~$i$ itself is correct or not, it reasons assuming that its own perceptions are the objective reality. The $\correct{j}$ assumption is, therefore, necessary to anchor~$j$ to the same (allegedly) objective reality  contemplated by~$i$, even though $j$'s~access to the facts of this objective reality is generally different from~$i$'s. 
Note also that $j$'s~reasoning is generally happening in the future relative to $i$'s~current reasoning, meaning that we also implicitly assume reality to be stable. 
\end{remarque}

\begin{remarque}[Relation to indexical sets]
Another approach to describing beliefs of  fault-prone agents is via  so-called indexical sets~\cite{bookof4,MT88}, which are variable (non-rigid) sets that can be used to represent the set of all correct agents at every point in the system. While our results could be reformulated in terms of indexical sets, there were several reasons for us to choose another language. Besides the ability to reason about all agents, whether correct or faulty, in a uniform way, we tried to stay as close as possible to the standard language of epistemic modal logic. Perhaps more importantly, however, was the moral lesson of the already mentioned Knowledge of Preconditions Principle~\cite{Mos15TARK}, which reveals how important it is for an agent to know all ingredients affecting its behavior, correctness of itself and other agents being one of them. Thus, we believe that the transparent and explicit use of correctness in our language is advantageous. An immediate example is the distinction between belief and hope discussed in Remark~\ref{rem:hope_orig}, which would have remained somewhat obscured in the indexical set notation.
\end{remarque}

\begin{remarque}[Eventual mutual hope is not sufficient]\label{rem:sufficientknowledge}
While using  $B_i\left(\truestart[] \wedge {E^{\eventually H}\truestart[]}\right)$ as a trigger for agent~$i$ firing  will ensure Unforgeability\/~\textup{(U)}, it is too weak to guarantee Relay\/~\textup{(R)}. Indeed, consider a system with 3~agents~($n=3$), at most one of which can become faulty~($f=1$). In such a system, receiving the same information from two independent sources is sufficient to believe in its validity, while  information from only one source  without observing it first hand is not. Suppose that the protocol forces a correct agent to notify all other agents whenever it observed\/ $\START$. 
Consider a run where agent~$b$ is byzantine from the beginning, whereas agents~$c_1$~and~$c_2$ remain correct. Let $c_1$~and~$c_2$~each observe\/ $\START$ and, hence, notify all agents  about it. Meanwhile $b$~falsely notifies~$c_2$ that it too observed\/ $\START$ but will never duplicate this message to~$c_1$. Thus, 
\begin{itemize}
\item correct~$c_2$ observed\/ $\START$ and eventually received 2~confirmations of\/ $\START$ from~$c_1$~and~$b$; \looseness=-1
\item correct~$c_1$ observed\/ $\START$ and eventually received 1~confirmation of\/ $\START$  from~$c_2$;
\item faulty~$b$ did not observe\/ $\START$ but was eventually notified of\/ $\START$ by both~$c_1$~and~$c_2$.
\end{itemize}
In this situation, all agents eventually believe that\/ $\START$ was correctly observed ($c_1$~and~$c_2$~saw it themselves, whereas $b$~has 2~independent confirmations). Moreover, $c_2$~has a reason to believe in the eventual mutual hope of\/ $\START$. 
Indeed, hope would be trivially satisfied for a faulty agent, whereas any correct agent would eventually receive at least 2~confirmations out of~3 that $c_2$~itself possesses. Thus, according to the proposed knowledge threshold,  $c_2$~should fire. 
On the other hand, $c_1$~will never fire because it cannot be sure that $b$~will eventually hope that\/ $\START$ occurred. In $c_1$'s~mind, if $b$~were correct and $c_2$~were faulty and did not send a confirmation to~$b$, then $b$~would only ever receive 1~confirmation, which is not sufficient to make it trust\/ $\START$ truly occurred. Hence, $c_1$~would never fire, and Relay\/~\textup{(R)} would be violated.

The issue here is that  
$\B{i}{E^{\eventually H}{\truestart[]}}$ for one correct agent~$i$ 
does not generally imply that eventually 
$\B{j}{E^{\eventually H}{\truestart[]}}$ for all other correct agents~$j$. 
\end{remarque}

Thus, although $\B{i}{E^{\eventually H}{\truestart[]}}$ is necessary before $i$~can fire and is in principle actionable, acting on it may be premature. The necessary state of knowledge must be further strengthened. 
Since FRR~involves an agreement property (one correct agent fires only if all other correct agents also fire eventually), it is not very surprising that, in fact, some form of common knowledge, specifically  \emph{eventual common hope}, plays a role.  We have shown (see a proof in the Appendix on p.~\pageref{proof:theoremfourteen}) that 
Unforgeability and Relay together imply that, in order to fire an agent must ascertain (modulo its own correctness) both that $\START$ was observed by some correct agent and  the eventual common hope of the same fact:
\begin{theorem}[State of knowledge necessary for firing in presence of both~(U)~and~(R)]\label{thm:necessarycommonFRR}
Let $I$~be an interpreted system consistent with\/~\textup{(U)}\/~and\/~\textup{(R)}. For any agent~$i\in\agents$, 
\begin{equation}
\label{eq:URR:nec}
I\quad\models\quad \truefire[i] \to \B{i}{\left(
\truestart[] \wedge {\CH{}{\truestart[]}}\right)
}.
\end{equation}
\end{theorem}

\begin{corollary}
\label{cor:nec_FRR}
For any  interpreted system consistent with~FRR,\/ \eqref{eq:URR:nec}~is satisfied for all agents.
\end{corollary}

We now show (see a proof in the Appendix on p.~\pageref{proof:theoremsixteen}) that, unlike belief in eventual mutual hope (see Remark~\ref{rem:sufficientknowledge}),  belief in eventual \emph{common} hope is sufficient to fulfill Unforgeability and Relay, i.e.,~that firing as soon as the necessary state of knowledge 
from Theorem~\ref{thm:necessarycommonFRR} is achieved does guarantee that both~(U)~and~(R) are fulfilled:
\begin{theorem}[Sufficient conditions for~(U)~and~(R)]\label{thm:necessaryandsufficientFR}
For any interpreted system~$\intsys$:
\begin{enumerate}
\item \textup{(U)}\/~is fulfilled if \qquad $I\quad \models\quad \bigwedge\nolimits_{i\in \agents}(\lnot \B{i}{\truestart[]} \to \lnot\truefire[i])$.
\item Both\/~\textup{(U)}\/~and~\/~\textup{(R)}\/~are fulfilled if 
\begin{equation}
\label{eq:sufFRR}
 I \models \bigwedge\limits_{i\in \agents}\Bigl(\Bigl(\lnot \B{i}{\left(
\truestart[] \wedge {\CH{}{\truestart[]}}\right)
} \to \lnot\truefire[i]\Bigr) \land \Bigl(\B{i}{\left(
\truestart[] \wedge {\CH{}{\truestart[]}}\right)
} \to \eventually{(\correct{i}\to\truefire[i]})\Bigr)\Bigr).
\end{equation}
\end{enumerate}
\end{theorem}

\begin{remarque}[Belief in $\truestart$ is not redundant]
Since common knowledge is the strongest type of knowledge and knowledge is supposed to be factive,
it might be tempting to think that  the conjunct $\truestart[]$ is redundant in the formulations of Theorems~\ref{thm:necessaryFR}, \ref{thm:necessarycommonFRR},~and~\ref{thm:necessaryandsufficientFR}.2. The difference in our setting is that the relevant epistemic state is \emph{eventual}, meaning that it need not be factual at present. Still, one might question how it could be possible to achieve even an eventual knowledge/belief/hope without the event actually happening. Indeed, if there is no reason for agents to expect\/ $\START$ to necessarily occur, their predictions about\/ $\START$  occurring  can only rely on it already having occurred. This observation is formalized in Lemma~\ref{lem:earlylocalknow} and the immediately following corollary.
\end{remarque}

\begin{definition}[Potentially persistent formulas]
A formula~$\varphi$ is called \emph{potentially persistent} in an interpreted system $I=(R,\pi)$ if, for any run~$r\in R$ and any time~$t\in\bbbt$ such that\/ $(I,r,t) \models \varphi$, there exists a run~$r' \in R$ such that $r'(t) = r(t)$ --- i.e.,~$r'$~is an alternative continuation of the global state~$r(t)$ --- and such that\/ $(I,r',t) \models \always{\varphi}$. In other words, a true potentially persistent formula can stay true forever.
\end{definition}

The following 3~lemmas follow from definitions  (proofs of the last two are   in the Appendix on p.~\pageref{proof:lemmanineteen}).\looseness=-1
\begin{lemma}
\label{lem:stability}
$I \models  \lnot \correct{i} \to \always{\lnot \correct{i}}$ for any~$i \in \agents$ and  interpreted system~$I$.
\end{lemma}

\begin{lemma}
\label{lem:resil}
$I \models K_i \eventually \lnot \varphi \to K_i \lnot\varphi$ for any~$i \in \agents$ and $\varphi$~potentially persistent in an interpreted system~$I$.
\end{lemma}

\begin{lemma}
\label{lem:beldia_knowdia}
$
I  \models  B_i\eventually{(\correct{i} \to \varphi)} \leftrightarrow K_i\eventually{(\correct{i} \to \varphi)}$\/  for any~$i \in \agents$, formula~$\varphi$, and  interpreted system~$I$, i.e.,~believing something eventually happens modulo one's own correctness is as strong as knowing it eventually happens modulo one's own correctness.
\end{lemma}
\begin{corollary}
\label{cor:beldiahop_knowdiahop}
$I \models  B_i\eventually{H_i\varphi} \leftrightarrow K_i\eventually{H_i \varphi}$ for any~$i \in \agents$, formula~$\varphi$, and  interpreted system~$I$.
\end{corollary}

\begin{lemma}[Early local belief]
\label{lem:earlylocalknow}
If $\correct{i} \land\lnot\truestart$ is potentially persistent in an interpreted system~$I$, 
\begin{equation*}
I\qquad \models \qquad B_i \eventually{H_i\truestart[]} \to B_i \truestart.
\end{equation*}
\end{lemma}
\begin{proof}
A proof can be found in the Appendix on p.~\pageref{proof:lemmatwentythree}.
\end{proof}

Noting that $I \models \CH{}{\truestart[]} \to E^{\eventually H}{\truestart[]}$ because of the normality of the hope modality~\cite{Fru19:ESSLLI}, we can derive:
\begin{corollary}
If $\correct{i} \land \lnot \truestart$ is potentially persistent in an interpreted system~$I$, then
\begin{gather}
\label{eq:earlyknow_mut}
I \qquad \models \qquad B_i E^{\eventually H}{\truestart[]}\to B_i \truestart,
\\
\label{eq:earlyknow_com}
I \qquad \models \qquad B_i \CH{}{\truestart[]}\to B_i \truestart.
\end{gather}
\end{corollary}

\begin{remarque}[Conditions on dropping~$B_i\truestart$]
\label{rem:nomuteruns}
If, contrary to the conditions of the early local belief lemma, $\truestart$ is inevitable, agents    \emph{may} be able to predict the eventual arrival of\/ $\START$ before the fact. For instance, if\/ $\START$ eventually happens to every agent, i.e.,~if $I \models \eventually{\trueoccurred[i]{\START}}$, then $I \models \eventually{B_i\truestart[]}$ for every agent~$i \in \agents$. 
It follows that $I \models C \CH{}{\truestart[]}$, i.e.,~it is common knowledge, from the very beginning, that there is eventual common hope of $\truestart[]$.  Thus, this state of knowledge is achieved independently of\/ $\START$ happening, and triggering\/ $\FIRE$ risks violating Unforgeability\/~\textup{(U)}.

On the other hand, even if\/ $\START$ is assured, it may not always be possible to predict it in advance. While sufficient for dropping the conjunct $\truestart$ from the conditions triggering\/ $\FIRE$ in Theorem~\ref{thm:necessaryandsufficientFR}.2, the potential persistency of $\correct{i} \land \lnot \truestart$ is  not necessary. Indeed, \eqref{eq:earlyknow_com}~can hold even when\/ $\START$ is always guaranteed to happen. For instance, in an interpreted system  where\/ $\START$ happens exactly once per run, no agent ever becomes faulty, and, in addition, agents never communicate, $I \models \lnot B_iE^{\eventually H}{\truestart[]}$ because only the agent who observed\/ $\START$ can learn that it already occurred. All the others can only be sure that\/ $\START$ will occur eventually. 
By~\eqref{eq:fixpointaxiom}, also $I \models \lnot B_iC^{\eventually H}{\truestart[]}$. Thus,  both implications~\eqref{eq:earlyknow_mut}~and~\eqref{eq:earlyknow_com} are vacuously true, allowing to drop $\truestart$ without affecting the behavior of agents, though admittedly in such interpreted systems agents should never fire anyways.
\end{remarque}

The following ``Lifting Lemma'' shows that Correctness~(C) lifts eventual mutual hope to eventual common hope. This way, the arbitrarily deep nested hope implied by the latter effectively collapses, a phenomenon that has also been reported for other problems~\cite{BM10:DISC}. A proof of the lemma can be found in the Appendix on p.~\pageref{proof:lemmatwentysix}.

\begin{lemma}[Lifting Lemma]\label{lem:allknow}
Let $I$~be an interpreted system consistent with\/~\textup{(C)} and  let\/ $|\agents| \geq 3f+1$, where $f\geq 0$~is the maximum number of byzantine faulty agents in a run. Furthermore, assume that
\begin{equation}
\label{eq:nec_UR_hope}
I \qquad\models\qquad \truefire[i] \to  B_i\left(\truestart[] \wedge {E^{\eventually H}\truestart[]}\right)
\end{equation}
holds. Then,  
\begin{equation}
\label{eq:lifting_main}
I  \qquad \models \qquad E^{\eventually H}_{}{\truestart[]} \to {\CH{}{\truestart[]}}.
\end{equation}
\end{lemma} 
\begin{corollary}
\label{cor:relay_through_correctness}
Let $I$~be an interpreted system with at least $3f+1$~agents.
If  $I$~is  consistent with~FRR, then~\eqref{eq:lifting_main} holds. 
\end{corollary}
\begin{proof}
In interpreted systems consistent with~(U)~and~(R), property~\eqref{eq:nec_UR_hope} follows from Theorem~\ref{thm:necessaryFR}. 
\end{proof}

\section{Conclusions and Future Work}
\label{sec:conclusions}

We introduced a problem called Firing Rebels with Relay~(FRR) and its weaker variant called Firing Rebels~(FR), which capture the essentials of a well-known building block for byzantine fault-tolerant distributed algorithms. The main purpose of our paper was to determine the state of knowledge correct agents must achieve in order to act~($\FIRE$) according to the specification of the problem at hand. Through a detailed epistemic analysis, we established that the necessary and sufficient levels of knowledge required for acting  rely on the novel notion of eventual common hope. We also found the conditions under which a single level of eventual mutual hope can  guarantee infinitely many levels of eventual common hope and explored the surprisingly non-trivial relationship of the eventual common hope of $\START$ with the actual appearance of $\START$.\looseness=-1

\noindent Regarding future work, our next step is to complete the characterization of (eventual)~common hope. More precisely, what remains to be done is developing an independent axiomatization of the (eventual)~common hope modality based on our existing axiomatization of the hope modality (which does not depend on the knowledge modality). In addition,
we are working on identifying necessary and sufficient communication structures and optimal
protocols for~FRR.

\bibliographystyle{eptcs}

\appendix
\section*{Appendix}
 
\begin{proof}[Proof of Lemma~\ref{lem:aux}]
\label{proof:lemmasix}

The argument is the same for $\FIRE$ and $\START$. We only provide it for the former. Let $I = (R,\pi)$. Consider a run $r \in \system{}$ and a node $(i,t) \in \points$. Assume $\kstruct{}{r}{t} \models \truefire[i]$. Since $i$~has perfect recall and this was a correct $\FIRE$ action, it was recorded and still remains in $i$'s~local history~$r(t)$. Consider  any  $r'\in \System$ and  $t' \in \mathbb{N}$ such that $\run{i}{t}=\run[']{i}{t'}$. Then  $\run[']{i}{t'}$~also contains a record of $\FIRE$.
If $\kstruct{}{r'}{t'} \models \correct{i}$, this record must correspond to a correct action and, consequently, $\kstruct{}{r'}{t'} \models \truefire[i]$. Since $\kstruct{}{r'}{t'} \models  \correct{i} \to \truefire[i]$ whenever $\run{i}{t}=\run[']{i}{t'}$, we have $\kstruct{}{r}{t}\models \K{i}{(\correct{i} \rightarrow \truefire[i])}$, i.e.,~$\kstruct{}{r}{t}\models B_i\truefire[i]$. The other statement about FIRE follows from $\models \truefire[i] \to \truefire$ and the monotonicity/normality of~$B_i$.\looseness=-1
\end{proof}

\smallskip

\begin{proof}[Proof of Theorem~\ref{thm:necessaryFR}]
\label{proof:theoremten} 
Since the system is consistent with~(U), \eqref{eq:U:nec}~holds according to Lemma~\ref{lem:U}. Thus, given that $B_i$~is a normal modality, it only remains to show that 
\begin{equation}
\label{eq:nec_FRR_second}
 I\qquad \models\qquad \truefire[i] \to  B_i {E^{\eventually H}\truestart[]}.
 \end{equation}
Since the system is consistent with~(R), \eqref{eq:R:nec}~holds according to Lemma~\ref{lem:R}.  Using the replacement property for positive subformulas and the already discussed validity $I \models \truefire[j] \to B_j \truestart[]$ from~\eqref{eq:U:nec}, we obtain $I \models \truefire[i] \to B_i{\bigwedge\nolimits_{j\in\agents} \eventually{\bigl(\correct{j} \rightarrow B_j\truestart[]\bigr)}}$, in other words,~\eqref{eq:nec_FRR_second}.
\end{proof}

\smallskip

\begin{proof}[Proof of Theorem~\ref{thm:necessarycommonFRR}]
\label{proof:theoremfourteen}
Since \eqref{eq:U:nec}~holds by Lemma~\ref{lem:U}, it is sufficient to demonstrate 
$I\models \truefire[i] \to  B_i{\CH{}{\truestart[]}}$. Combining~(R)~with~\eqref{eq:aux} by applying the replacement property for positive subformulas, we obtain
$
I \models \truefire \to {E^{\eventually H}{\truefire}}
$.
Thus, using~\eqref{eq:inductionrule} with $\varphi=\psi = \truefire$, we conclude $
I \models \truefire \to  \CH{}{\truefire}$.
Since the greatest fixed point of a monotone operator is itself monotone, it follows from~(U) that 
$
I \models \truefire \to  \CH{}{\truestart}$.
It remains to use~\eqref{eq:aux} and monotonicity of~$B_i$ to obtain
$
I \models \truefire[i] \to B_i\CH{}{\truestart}$.
\end{proof}

\smallskip

\begin{proof}[Proof of Theorem~\ref{thm:necessaryandsufficientFR}]
\label{proof:theoremsixteen}
For either assumption,  $I \models \truefire[i] \to \B{i}{\truestart}$. Since 
\begin{equation}
\label{eq:fire_fact}
I \quad\models\quad \truefire[i] \to \correct{i}
\qquad \text{and} \qquad I \quad\models\quad \correct{i} \rightarrow (\B{i}{\varphi}\rightarrow\varphi) \quad \text{for any formula~$\varphi$},
\end{equation}
we have
$I \models \truefire[i] \to \truestart$ for each $i \in\agents$. Since $\truefire$ is  $\bigvee_{i \in \agents} \truefire[i]$,  (U)~holds by propositional reasoning.\looseness=-1

It remains to show that Relay~(R) holds under the assumption of~\eqref{eq:sufFRR}. Once again, it is sufficient to demonstrate that, for each $i \in \agents$,
\begin{equation}
\label{eq:R_what_we_want}
I \qquad\models\qquad \truefire[i] \to \bigwedge\nolimits_{j\in \agents} \eventually (\correct{j}\to \truefire[j]).
\end{equation}
It follows from the first conjunct of~\eqref{eq:sufFRR} that 
$I \models \truefire[i] \to  \B{i}{
 {\CH{}{\truestart[]}}
}$. Using~\eqref{eq:fire_fact} again, we conclude that $I \models \truefire[i] \to  
 {\CH{}{\truestart[]}
}$.
Since $I \models \CH{}{\varphi} \to \bigwedge_{j\in\agents}\eventually{H_j(\varphi \land \CH{}{\varphi})}$ for any formula~$\varphi$ according to~\eqref{eq:fixpointaxiom}, 
\begin{equation}
I \qquad\models\qquad  \truefire[i] \to  \bigwedge\nolimits_{j\in\agents}\eventually{\Bigl(\correct{j} \to \B{j}{\bigl(\truestart[] \wedge {\CH{}{\truestart[]}}}\bigr)\Bigr)}.\label{eq:intermed}
\end{equation}
Using the second conjunct  of~\eqref{eq:sufFRR} and monotonicity of~$B_j$~and~$\eventually{}$ in~\eqref{eq:intermed}, we obtain
\[
I\qquad \models \qquad \truefire[i] \to  \bigwedge\nolimits_{j\in\agents}\eventually{\Bigl(\correct{j} \to \eventually{\bigl(\correct{j}\to\truefire[j]}\bigr)\Bigr)}.
\]
To get~\eqref{eq:R_what_we_want}, it remains to note that
$I \models \eventually{\Bigl(\varphi \to \eventually{(\varphi \to \psi)}\Bigr)}\to \eventually{(\varphi \to \psi)}$ for all formulas~$\varphi$~and~$\psi$.
\end{proof}

\smallskip

\begin{proof}[Proof of Lemma~\ref{lem:resil}]
\label{proof:lemmanineteen}
Let $I=(R,\pi)$. 
Assume that $(I,r,t) \not\models  K_i \lnot \varphi$ for some $r \in R$ and $t \in \bbbt$. Then there exists another run $r'\in R$ and time $t'\in\bbbt$ such that $r_i(t)= r'_i(t')$ and $(I,r',t') \models \varphi$. By the potential persistence of~$\varphi$, there exists an alternative continuation $r''\in R$ of the prefix $r'(t')$ such that $r''(t')=r'(t')$ and $(I,r'',t') \models \always{\varphi}$. Thus, $(I,r'',t') \not\models \eventually{\lnot\varphi}$. It remains to note that $r''_i(t')=r_i'(t')=r_i(t)$. Hence, $(I,r,t) \not \models K_i\eventually{\lnot\varphi}$.
\end{proof}

\smallskip

\begin{proof}[Proof of Lemma~\ref{lem:beldia_knowdia}]
\label{proof:lemmatwentyone}
The right-to-left direction is trivial. Hence, we prove the implication from left to right.
Firstly, $\lnot \correct{i} \to (\correct{i} \to \varphi)$  is a propositional tautology. Hence, 
\[
I \qquad\models\qquad \always{\lnot \correct{i}} \to \always{(\correct{i} \to \varphi)}.
\] 
Using Lemma~\ref{lem:stability}, the fact that  $I \models \always{\psi} \to \eventually{\psi}$ by seriality of temporal modalities, and knowledge necessitation, we obtain 
\[I \qquad\models\qquad K_i\bigl(\lnot \correct{i} \to \eventually{(\correct{i} \to \varphi)}\bigr).
\]
By epistemically internalized propositional reasoning, 
\[
I \quad\models\quad K_i\bigl( \correct{i} \to \eventually{(\correct{i} \to \varphi)}\bigr) \land K_i\bigl(\lnot \correct{i} \to \eventually{(\correct{i} \to \varphi)}\bigr) \to K_i\eventually{(\correct{i} \to \varphi)}.
\]
Since we have just shown the second conjunct above to be valid, we obtain the desired
\[
I \qquad\models\qquad  K_i\bigl( \correct{i} \to \eventually{(\correct{i} \to \varphi)}\bigr) \to K_i\eventually{(\correct{i} \to \varphi)}.\qedhere
\]
\end{proof}

\smallskip

\begin{proof}[Proof of Lemma~\ref{lem:earlylocalknow}]
\label{proof:lemmatwentythree}
By Corollary~\ref{cor:beldiahop_knowdiahop}, $I \models B_i \eventually{H_i\truestart[]} \to K_i\eventually{H_i\truestart[]}$.
Applying factivity of knowledge and propositional reasoning to the expanded version of $K_i\eventually{H_i\truestart[]}$ yields 
\[
I \qquad \models \qquad K_i\eventually{\bigl(\correct{i} \to K_i(\correct{i}\to\truestart[])\bigr)} \to K_i\eventually{(\correct{i} \to  \truestart[])}.
\]
Since $\correct{i} \land \lnot \truestart$ is potentially persistent, and its negation is equivalent to $\correct{i} \to \truestart$, we have by Lemma~\ref{lem:resil} that
\[
I \qquad \models \qquad K_i\eventually{(\correct{i} \to  \truestart[])} \to  K_i {(\correct{i} \to  \truestart[])} 
\]
Combining all implications together, we conclude that
$
I \models  B_i \eventually{H_i\truestart[]} \to B_i\truestart$.
\end{proof}

\smallskip

\begin{proof}[Proof of Lemma~\ref{lem:allknow}]
\label{proof:lemmatwentysix}
Let $I=(R,\pi)$. 
Assume $\kstruct{}{r}{t} \models E^{\eventually H}_{}{\truestart[]}$ for some $r \in R$ and time $t\in\bbbt$. 
This means that,  for every agent $j\in \agents$, there is some $t_{j}' \geq t$ such that $\kstruct{}{r}{t_{j}'} \models \H{j}{\truestart[]}$. 
Since $|\agents| \geq 3f+1$, it follows that there exists a group~$G$ of $2f+1$~correct agents such that 
$\kstruct{}{r}{t_{j}'} \models \H{j}{\truestart[]}$ for all $j \in G$.  
Since these agents are correct,\footnote{While we only use the fact that agent~$j \in G$ is correct at~$t'_j$, these agents will necessarily remain correct throughout run~$r$.} we have $\kstruct{}{r}{t_{j}'} \models \B{j}{\truestart[]}$, i.e., $\kstruct{}{r}{t_{j}'} \models \K{j}({\correct{j} \rightarrow \truestart[]})$ for all  $j\in G$. Let $t' \ce \max\{t'_j \mid j \in G\}$. 
We claim that 
\begin{equation}
\label{eq:lifting}
\kstruct{}{r}{t'} \qquad\models\qquad \bigwedge\nolimits_{j \in G} \K{j}({\correct{j} \rightarrow \truestart[]}).
\end{equation}
Indeed, for any agent $j \in G$ consider any alternative run $\overline{r}\in R$ and time $\overline{t'}\in\bbbt$ such that $\overline{r}_j(\overline{t'}) = r_j(t')$. 
Given that $t' \geq t'_j$ and our agents have perfect recall, there must exist some time $\overline{t'_j}\leq \overline{t'}$ such that \mbox{$\overline{r}_j(\overline{t'_j}) = r_j(t'_j)$}. 
Thus, $\kstruct{}{\overline{r}}{\overline{t'_j}} \models \correct{j} \rightarrow \truestart[]$. 
Since the latter formula is stable, it remains true in~$\overline{r}$ by the time~$\overline{t'}$. We  showed that $\kstruct{}{\overline{r}}{\overline{t'}} \models \correct{j} \rightarrow \truestart[]$ whenever $\overline{r}_j(\overline{t'}) = r_j(t')$, meaning $\kstruct{}{r}{t'} \models  \K{j}({\correct{j} \rightarrow \truestart[]})$. This argument applies to every $j \in G$, hence, \eqref{eq:lifting}~is demonstrated for the group~$G$ of $2f+1$~correct agents.

Correctness~(C) applied to~$G$ at time~$t'$ ensures 
$
\kstruct{}{r}{t'} \models \bigwedge_{i \in \agents}\eventually (\correct{i} \to \truefire[i])
$,
and, since $t \leq t'$, we also have 
\[
\kstruct{}{r}{t}\qquad \models\qquad \bigwedge\nolimits_{i \in \agents}\eventually (\correct{i} \to \truefire[i]).
\] 
Given that $r$~and~$t$~were chosen arbitrarily, we have proved
\[
I \qquad\models\qquad E^{\eventually H}_{}{\truestart[]} \to \bigwedge\nolimits_{i \in \agents}\eventually (\correct{i} \to \truefire[i]).
\]
Using~\eqref{eq:nec_UR_hope}, we can conclude 
\[
I \qquad\models\qquad E^{\eventually H}_{}{\truestart[]} \to  \bigwedge\nolimits_{i \in \agents}\eventually{\bigl(\correct{i} \to
\B{i}{(\truestart[] \land {E^{\eventually H}_{}{\truestart[]})}}\bigr)},\] 
i.e., 
\[
I\qquad \models\qquad E^{\eventually H}_{}{\truestart[]} \to \bigwedge\nolimits_{i \in \agents} \eventually H_{i}{(\truestart[] \land {E^{\eventually H}_{}{\truestart[]})}}.
\] 
In other words, we have demonstrated
\[
I \qquad\models \qquad E^{\eventually H}_{}{\truestart[]} \to E^{\eventually H}_{}{(\truestart[] \land E^{\eventually H}_{}{\truestart[]})}.
\]
Using~\eqref{eq:inductionrule} with $\psi =  E^{\eventually H}_{}{\truestart[]}$ and $\varphi = \truestart[]$, we conclude 
\[
I \qquad\models \qquad E^{\eventually H}_{}{\truestart[]}\to  {\CH{}{\truestart[]}}.\qedhere
\]
\end{proof}
\end{document}